\documentclass[11pt,a4paper]{article}

\usepackage[utf8]{inputenc}
\usepackage[T1]{fontenc}
\usepackage{lmodern}
\usepackage{amsmath}
\usepackage{amssymb}
\usepackage{amsthm}
\usepackage{booktabs}
\usepackage{graphicx}
\usepackage{longtable}
\usepackage{array}
\usepackage{float}
\usepackage{microtype}
\usepackage[margin=1in]{geometry}
\usepackage[colorlinks=true,linkcolor=blue,citecolor=blue,urlcolor=blue]{hyperref}

\theoremstyle{definition}
\newtheorem{proposition}{Proposition}
\newtheorem{lemma}{Lemma}
\newtheorem{theorem}{Theorem}
\newtheorem{corollary}{Corollary}
\newtheorem{definition}{Definition}
\newtheorem*{definitionstar}{Definition}
\newtheorem*{remark}{Remark}
\makeatletter
\def\thmhead@plain#1#2#3{%
  \thmname{#1}\thmnumber{\@ifnotempty{#1}{ }\@upn{#2}}%
  \thmnote{ {\the\thm@notefont#3}}}
\let\thmhead\thmhead@plain
\makeatother

\newcommand{\proofstatus}[1]{[\textsc{PROOF-STATUS}: \textsc{#1}]}

\newcommand{\rhosub}{\rho_{\mathrm{sub}}}
\newcommand{\rhoroute}{\rho_{\mathrm{route}}}

\title{One Gate Is Not Enough: Composing Stateful Pre-Action Controls for Agentic AI}
\author{Gaston Besanson\thanks{Universidad Torcuato Di Tella}}
\date{}

\begin{document}
\maketitle

\noindent Companion artifact: \texttt{suite-demo} (open data, deterministic, Apache 2.0).

\begin{abstract}
Agentic AI systems take consequential actions governed by more than one concern at once: is the agent permitted to act, can the organisation afford the action, and is the evidence behind it valid. Prior work treats these as separate pre-action gates: SARC for obligations and permissions (arXiv 2605.07728), Green SARC for predictive cost and carbon budgets (arXiv 2606.15954), and SARC-DQ for metadata-borne evidence validity (arXiv 2607.26313). Prior research studies composition of enforcement monitors and policy decisions; this paper studies a narrower problem that arises when one control can remediate the action, evidence, or derived context consumed by another control. We formalize remediation-induced control coupling, in which a control transformation invalidates another control's earlier judgment. We show that naive single-pass composition can consequently be unsound, and give a remediate-and-regate protocol that restores per-action soundness in the current bounded, idempotent setting under its stated assumptions. We further show that the two implemented remediation operators, evidence substitution and resource-budget downroute, do not commute -{}- a finite-model checker finds concrete counterexample instances -{}- making remediation order part of the control-plane semantics rather than an implementation detail. At the state level, we demonstrate that currently admissible observations can contaminate future governance state when uncovered defects are promoted into a governed evidence buffer; two mitigations reduce, not eliminate, that exposure. Supporting results establish the exact condition under which positive-weight linear aggregation of gate outcomes can compensate a member veto (a vetoed action is admitted exactly when the admission threshold does not exceed the aggregator's largest leave-one-out weight sum), a unified cross-control Evidence Set preserving full lineage across gates, and the fact that composition manufactures no new detection coverage: classes no member gate covers remain uncovered, reported honestly. A deterministic artifact composing the three published engines unmodified validates the mechanisms over 30 pre-registered seeds and two workflows: CH1-CH5 meet their registered decision rules across all 30 pre-registered seeds; CH6 does so under W1 but not under the smaller W2 workflow, reported as such. The artifact is a mechanism demonstration, not a production-prevalence study.
\end{abstract}

\section{Introduction}

A control that changes the action it governs can invalidate the decisions of controls that evaluated the pre-remediation action. A pre-action gate is a deterministic checkpoint between an agent's decision and its execution. Three families of pre-action concern are separately established: authority, resources, and evidence. Deployed systems need all three simultaneously.

Prior research has studied the composition of enforcement monitors and policy decisions: serial and parallel monitor composition, policy-combining algorithms (deny-overrides and related patterns), non-compensatory veto rules in multi-criteria decision analysis, stateful runtime enforcement, and the confluence and termination theory of rewrite systems are all established literatures. This paper does not claim to invent the composition of pre-action controls, and does not claim that the composition question itself is untreated. What remains underexplored is the composition of heterogeneous, stateful pre-action controls in which one control's remediation mutates the action, evidence, or derived context evaluated by another control, thereby invalidating previously computed judgments and creating order-dependent behavior. That distinction -{}- heterogeneous control semantics, different control-local state, action mutation, evidence mutation, derived-context recomputation, budget state, governance-buffer state, remediation ordering, cross-control lineage, and execution-time re-evaluation, together -{}- is this paper's actual subject, narrower than a claim that composition per se is untreated, and more defensible.

The question is not academic. The evidence gate's designed remediation, quarantine-and-substitute from a governed buffer, changes the acted-on value. A substituted unit cost changes the order quantity a replenishment agent proposes, which changes the order value the authority cap must judge and the predicted spend the resource gate must budget. A plane that evaluates all gates once, in parallel, on the original action is judging an action that will not be the one executed. Once a second remediator is introduced -{}- a resource gate that, rather than simply blocking an over-budget action, can scale its quantity down to the largest budget-feasible amount -{}- the same question recurs one level up: does it matter which remediator runs first.

\noindent\textbf{Contributions.} This paper makes four primary contributions. First, it formalizes remediation-induced control coupling: a pre-action control that transforms an action can invalidate another control's earlier judgment (Section 2, Section 4, scenario S4). Second, it gives and evaluates a remediate-regate composition protocol for the current idempotent setting, under which every executing action is re-evaluated by all member controls after remediation (Section 4, Theorem 1, CH1). Third, it shows that heterogeneous remediation operators need not commute, making remediation order part of the control-plane semantics rather than an implementation detail (Section 5, CH7). Fourth, it demonstrates that governance state itself can become contaminated when admissible but uncovered defects are promoted into future remediation state, motivating explicit trust policies for state updates (Section 6, CH6). Non-compensatory join semantics (Section 3, machine-checked exhaustively over the finite response lattice), a unified Evidence Set with cross-gate lineage preservation (Section 7), and a no-manufactured-coverage property verified by injecting two declared-uncovered classes (Section 8, CH4) support these four primary contributions. The empirical artifact underneath all of them is a deterministic open-data composition of the three published Apache 2.0 engines unmodified, with hypotheses CH1 to CH8 pre-registered before the corresponding code ran, and a 30-seed statistical sweep reporting means with 95\% confidence intervals (Section 8).

\noindent\textbf{Scope honesty.} This paper claims composition semantics and a mechanism demonstration. It does not claim production prevalence, uses no model calls, and its artifact is not a GIGO-Bench release.

\section{Model and terminology}

An action $a$ is proposed in context $c(a) = (\mathrm{agent}, \mathrm{role}, \mathrm{parameters}, \mathrm{value}, \mathrm{predicted\ resource\ use})$. An evidence set $E(a)$ is the finite sequence of records $a$ relies on; each record is a payload plus metadata (source, as-of, retrieval time, version, lineage) with a content-addressed identifier $\mathrm{eid}(r)$. A gate $G_i$ maps $(a, c, E, s_i)$ to a response in $R = \{\mathrm{admit}, \mathrm{substitute}, \mathrm{degrade}, \mathrm{escalate}, \mathrm{block}\}$, where $s_i$ is gate-local state (for the resource gate, remaining budgets; for the evidence gate, the governed buffer). Partition $R$ into $\mathrm{Exec} = \{\mathrm{admit}, \mathrm{substitute}, \mathrm{degrade}\}$ and $\mathrm{Held} = \{\mathrm{escalate}, \mathrm{block}\}$: this Exec/Held partition, not the five-level total order below, is the load-bearing distinction for execution safety, and is what a plane's soundness (Definition 3) is stated in terms of.

\setcounter{definition}{2}
\begin{definition}[(per-action soundness)]
A plane is sound for $a$ if the executed action $a_{\mathrm{exec}}$ satisfies every gate at execution time, on the state in force when $a_{\mathrm{exec}}$ runs.
\end{definition}

\begin{remark}[(sequential coupling)]
Resource-gate state decrements as actions execute; soundness here is per action given current state; stream-level ordering across actions is out of scope and flagged in Section 10.
\end{remark}

A remediation operator $\rho$ is a control's own transformation of an action on a violation it is designed to correct rather than merely block: it maps an action $a$ to a remediated action $\rho(a)$, typically also changing $c(\rho(a))$. Section 4 introduces the evidence gate's substitution operator $\rhosub$; Section 5 introduces the resource gate's downroute operator $\rhoroute$.

\begin{definitionstar}[(remediation-induced control coupling)]
Two controls $G_i$ and $G_j$ are remediation-coupled under remediation operator $\rho_i$ if there exists an executable-reachable action $a$ such that $G_j(a) \neq G_j(\rho_i(a))$.
\end{definitionstar}

\noindent This concept is reusable independently of SARC, Green SARC, or SARC-DQ: any pre-action control whose remediation changes a value another control reads exhibits it. Scenario S4 (Section 4) is this artifact's concrete witness.

\noindent\textbf{Coupling graph.} Define a directed graph $\mathcal{C} = (\mathcal{G}, \mathcal{E}_C)$ where $(G_i, G_j) \in \mathcal{E}_C$ iff there exists a reachable action $a$ such that $G_j(a) \neq G_j(\rho_i(a))$; this is the remediation-induced control coupling graph. For the current example: $DQ \to Authority$ and $DQ \to Resource$. A resource downroute can similarly affect values consumed by other controls.

\noindent\textbf{Design rule (judgment invalidation).} If remediation by control $G_i$ can change any input consumed by control $G_j$, the earlier judgment of $G_j$ must be treated as stale for the remediated action and recomputed before execution. When the dependency graph is incomplete or dynamic, full regating is the conservative strategy.

This paper distinguishes three progressively harder composition problems, in increasing order of difficulty.

\noindent\textbf{Constraint composition.} Multiple controls judge the same action without any of them transforming it. When controls represent hard requirements, they behave as constraints rather than compensable preferences: an action is executable only if it remains inside the intersection of the relevant feasible regions, $F = F_{\mathrm{authority}} \cap F_{\mathrm{resource}} \cap F_{\mathrm{evidence}}$. This is the conceptual interpretation for the non-compensatory join semantics of Section 3.

\noindent\textbf{Remediation-induced control coupling.} A remediation applied by one control can change variables another control consumes: the evidence gate substitutes unit cost; unit cost changes order value; order value changes the authority decision; unit cost also changes predicted spend; predicted spend changes the resource decision. Section 4 studies this problem for a single remediation operator.

\noindent\textbf{Remediator interaction.} When more than one control can transform the action, remediation operators can interact non-commutatively, making remediation order part of governance semantics rather than an implementation detail. Section 5 studies this problem; it is one of this paper's strongest results.

The dependency this creates is the paper's central intuitive argument: remediation changes inputs consumed by other controls.

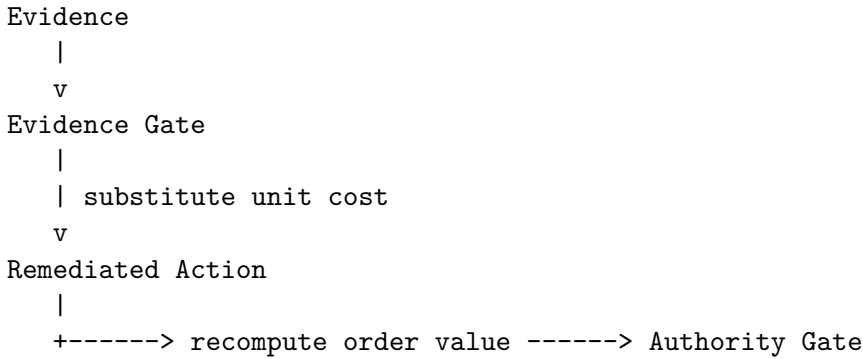
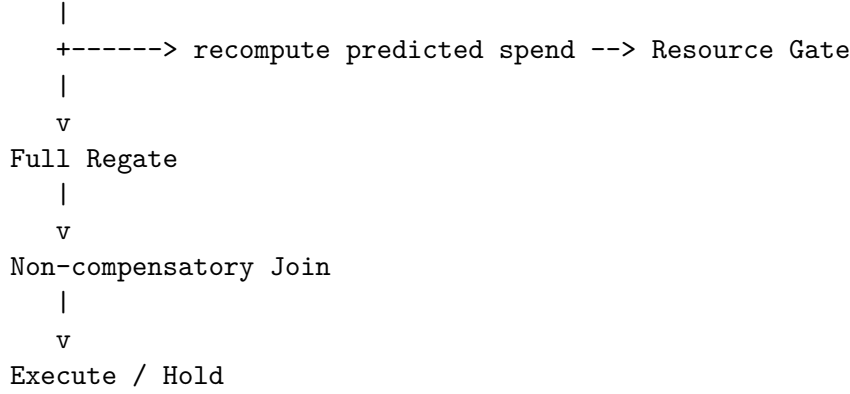
\begin{verbatim}
Evidence
   |
   v
Evidence Gate
   |
   | substitute unit cost
   v
Remediated Action
   |
   +------> recompute order value ------> Authority Gate
   |
   +------> recompute predicted spend --> Resource Gate
   |
   v
Full Regate
   |
   v
Non-compensatory Join
   |
   v
Execute / Hold
\end{verbatim}

\section{Hard constraints versus compensatory aggregation}

\setcounter{definition}{0}
\begin{definition}[(restrictiveness order)]
Order $R$ by permissiveness: $\mathrm{admit} \sqsubseteq \mathrm{substitute} \sqsubseteq \mathrm{degrade} \sqsubseteq \mathrm{escalate} \sqsubseteq \mathrm{block}$; $(R, \sqcup)$ is a join semilattice.
\end{definition}

\begin{definition}[(composed verdict)]
For responses $r_1..r_n$ on the same $(a, c, E)$, the composed response is the join $r^* = \sqcup_i r_i$, the composed admitted bit is $[r^*$ in $\mathrm{Exec}]$, and the winner gate is any argmax (ties recorded).
\end{definition}

Represent each gate's outcome as a score $s_i$ in $[0,1]$ with $s_i = 0$ iff the gate holds the action, and let an aggregator $f$ admit $a$ iff $f(s) \geq \tau$ with $\tau > 0$. An earlier draft of this paper claimed that any strictly increasing $f$ with some admissible profile violates veto -{}- an independent review refuted this: for $f(s) = s1+s2+s3$ and $\tau = 2.5$, $f(1,1,1) = 3$ is admissible, yet no profile with a coordinate at 0 can reach 2.5 (its maximum there is 2), so veto is fully preserved despite $f$ being strictly increasing and having an admissible profile. The universal form was wrong to treat those two hypotheses as sufficient; whether compensation is possible depends on the relationship between $\tau$ and $f$'s own weight structure. The deeper distinction this section formalizes is not merely min versus weighted sum: it is constraint satisfaction versus preference aggregation. Weighted, additive aggregation can be useful for selecting among already-feasible alternatives, but for hard governance requirements where member controls are meant to hold veto authority, additive preference aggregation is not a substitute for conjunctive feasibility semantics -{}- the linear-family lemma below gives the exact threshold condition for when a weighted-sum aggregator does or does not preserve that veto property, rather than assuming it either way. This paper is not arguing against weighted aggregation. It is arguing against using compensatory preference aggregation to implement controls whose semantics require independent veto authority.

\noindent\textbf{Lemma (linear family; general arbitrary-$n$ statement retagged pending-human-review, finite $n=3$ encoding split into Corollary 2, per independent review round two finding R2-N2).} Let $f(s) = \sum_i w_i s_i$ with $w_i \geq 0$, not all zero, and let $L_j = \sum_{i \neq j} w_i$ (the largest score achievable with $s_j = 0$), $L^* = \max_j L_j$. A vetoed profile ($s_j = 0$ for some $j$) is admitted by $f$ if and only if $\tau \leq L^*$. Min never admits a vetoed profile, for any $\tau > 0$, regardless of weights. Full proof, the reviewer's counterexample worked through against this exact condition, and the discrete instantiation's exhaustive machine check (including 200 HEAD-derived random weight-vector probes), in Appendix A.

Consequence: the composed verdict of Definition 2 is the ordinal form of min on admissibility, the cross-gate generalisation of ``any failed predicate blocks''. The linear-family lemma above is therefore best read as an exact characterization of when an additive aggregator does or does not preserve the declared veto semantics, not as a claim that additive scoring is inherently unsafe.

\noindent\textbf{CH5 (compensation is empirically unsafe, discrete grid).} Over the pre-registered weight/threshold grid (66 weight vectors x 4 thresholds, 264 cells), does any cell's weighted-score aggregator admit at least one of the 98 min-join-Held profiles out of the full 125-point three-gate verdict space, while the min join admits zero of them by construction. Result: \texttt{proposition\_1\_holds\_exhaustively = True}. Two distinct quantities are reported here, with distinct labels and distinct denominators, since they are not the same measurement: the \textbf{formal result} is the maximum number of compensated response profiles at a single grid cell, 48 out of 98 Held response profiles, from the exhaustive discrete checker (\path{checkers/compensation_check.py}); the \textbf{empirical result} is the mean number of compensated held decision instances pooled across scenarios S1-S4, over the 30-seed sweep: 13775.1 (95\% CI [13753.5, 13796.6], n=30) (\texttt{ch5\_aggregator.py}, \path{out/results/sweep_summary.json}). The two are not comparable by denominator (one counts formal response profiles out of the full discrete grid's Held population; the other counts individual decision instances actually produced across four scenarios' worth of simulated traffic), and this paper does not conflate them under one name.

\section{Single-remediator composition}

\noindent\textbf{Proposition 2.} If no gate rewrites $(a, c, E)$, the composed verdict is invariant to evaluation order and duplication, and adding a gate never increases permissiveness. Proof, and exhaustive machine check over the finite response lattice (17151 individual checks, all\_hold = True), in Appendix A.

The evidence gate is not read-only: on a substitutable violation it returns substitute with a governed value $v'$, defining $\rhosub(a) = a[v \to v']$, and $c(\rhosub(a))$ may differ from $c(a)$: order quantity, order value, and predicted spend can all move. This is remediation-induced control coupling (Section 2) in its simplest form: one remediation operator, applied once.

\noindent\textbf{Proposition 3 (single-pass unsoundness, scoped to the implemented construction).} For the substituting-DQ-decision construction below, single-pass evaluation on $(a, c(a), E(a))$ yields a composed Exec verdict whose executed action $\rhosub(a)$ violates the authority or resource gate at execution time; and symmetrically, single-pass holds an action whose remediated form is compliant. Construction: place an authority cap $\kappa$ strictly between the order values induced by the corrupted read and by the governed substitute. The artifact instantiates this as scenario S4 across all 30 registered seeds; the claim is not made for arbitrary continuous economics or arbitrary gate implementations beyond this construction (independent review classification: checked-scope-only). Full proof in Appendix A.

\noindent\textbf{Protocol (remediate-regate composition, ``rtr'' in artifact output).} Phase I: evaluate the evidence gate; on substitution form $a' = \rhosub(a)$ and recompute $c(a')$. Under W2, if the resource gate would otherwise reject $a'$ outright as over budget, apply $\rhoroute(a')$ and recompute $c(a')$ again. Phase II: evaluate every gate, including the evidence gate, on $(a', c(a'), E(a'))$; return the join of Phase II responses, recording Phase I in the Evidence Set.

\noindent\textbf{Lemma 1 (termination, scoped to the current one-shot composition branch).} For the current one-shot composition branch and the DQ-library behavior exercised in this artifact's scenarios/tests, buffer substitution is idempotent, $\rhosub(\rhosub(a)) = \rhosub(a)$, and $E(\rhosub(a))$ consists of governed records the evidence gate admits by construction; hence no further remediation is generated and the protocol reaches a fixed point after at most one remediation of each operator. The argument depends on an external DQ predicate contract this artifact composes but does not itself prove for every possible governed record (independent review classification: checked-scope-only). Full proof in Appendix A.

\noindent\textbf{Assumptions for the current composition theorem.} The soundness result below rests on the following assumptions, made explicit here as first-class objects rather than left implicit in the theorem's proof.

\begin{itemize}
\item \textbf{A1. Determinism.} Each member control is deterministic for fixed action, evidence, context, and control-local state.
\item \textbf{A2. Bounded remediation.} The remediation branch exercised by the artifact is one-shot or otherwise bounded.
\item \textbf{A3. Idempotence.} The evidence substitution operator is idempotent in the modeled branch.
\item \textbf{A4. Governed-record admissibility.} The governed substitute satisfies the evidence-predicate contract exercised by the artifact.
\item \textbf{A5. Derived-context recomputation.} Every derived value consumed by another control is recomputed after remediation.
\item \textbf{A6. Final regating.} Every applicable control is reevaluated on the final remediated action before execution.
\item \textbf{A7. No concurrent state mutation in theorem scope.} Relevant control state does not change between final evaluation and execution; stream-level concurrency is outside the theorem's scope.
\end{itemize}

\noindent Read together: $A_1 \land A_2 \land A_3 \land A_4 \land A_5 \land A_6 \land A_7 \Rightarrow \mathit{PerActionSoundness}$ (Definition 3).

\noindent\textbf{Theorem 1 (soundness of remediate-regate composition; general statement retagged checked-scope-only per independent review round two, finite response-lattice encoding split into Corollary 3).} Under Lemma 1 and gates that are functions of (action, context, evidence, current state), the remediate-regate protocol satisfies Definition 3: Phase II evaluates every gate on $a_{\mathrm{exec}}$ itself, and the join preserves every Held verdict. \textbf{Corollary 1 (sufficiency only; retagged checked-scope-only per independent review round two):} if the authority and resource gates are invariant under the operators applied on the executed-reachable action set, single-pass is sound. The converse does not hold in general -{}- an independent review gave a counterexample (a gate that varies only on actions independently held, and so never executed) -{}- and this paper no longer claims it; an earlier draft's iff overstated what was proved. Full proof, and the counterexample worked through, in Appendix A.

\section{Multi-remediator composition}

A second workflow (W2, weekly commitment) introduces a second remediation operator: $\rhoroute(a)$, which scales the committed quantity down to the largest quantity whose predicted cost and carbon both fit the remaining weekly budget, then recomputes context from the scaled quantity -{}- the same remediate-and-regate pattern Theorem 1 establishes for evidence substitution, applied to a different field of the action. With both operators active, does a fixed order matter -{}- does remediator interaction (Section 2) actually occur for this artifact's own two operators.

\noindent\textbf{CH7 (multi-remediator order dependence).} Is $\rhoroute(\rhosub(a))$ always equal to $\rhosub(\rhoroute(a))$? A finite-model checker (\texttt{checkers/remediator\_check.py}) applies both orders over a grid of (243 points): quantity, corrupted unit cost, governed unit cost, cost budget, carbon budget, reusing the real remediation functions directly. Registered outcome: \texttt{non\_confluent}, with 86 disagreeing grid points out of 243. The independent review additionally probed continuous-valued points off this pre-registered grid and found order-divergent outcomes there too (\texttt{review/REVIEW.md} Section 5: 144/200 HEAD-derived off-grid points order-divergent); this artifact promotes that probe into the checker itself and reruns it on every checker run, seeded from a fixed declared constant in \texttt{prereg/probe-seeds.json} rather than the current git HEAD, so the published count is commit-stable rather than drifting with every commit (round-two independent review finding R2-N1): 150/200 off-grid points disagree, strengthening CH7 beyond the declared grid. Where the two orders disagree, the mechanism is structurally the same failure Proposition 3 identifies between composition protocols: a budget check run against a not-yet-corrected value never gets re-checked once the value is corrected. This is the formal justification for \texttt{prereg/w2-workflow.md}'s fixed ordering (evidence gate first, then resource gate) -{}- not an arbitrary implementation choice but a consequence of the operators' non-commutativity. Full result in Appendix A.

\noindent\textbf{Design rule (remediation order).} When remediation operators do not commute, their ordering is part of governance policy and must be specified, versioned, and auditable rather than left to implementation accident.

\begin{verbatim}
rho_sub
  -> unit_cost
      -> order_value
          -> authority
      -> predicted_cost
          -> resource

rho_route
  -> quantity
      -> order_value
          -> authority
      -> predicted_cost
          -> resource
\end{verbatim}

Both remediators affect variables downstream controls consume, which is exactly why they can fail to commute. A fixed order is required for the two operators this artifact implements, and CH7 above is the formal justification for pre-registering it rather than leaving it unspecified. A fixed order is not, however, a general solution to arbitrary multi-remediator composition: with three or more operators, applying one can make an earlier one relevant again, and this artifact does not solve arbitrary non-idempotent remediators, cycles, global termination, unique fixed points, confluence across arbitrary operators, concurrent state updates, or serializability across multiple agents for that general case. This work establishes a sound one-shot composition protocol for the implemented idempotent setting and demonstrates why general multi-remediator systems require an explicit termination and confluence theory; see \texttt{docs/generalized-composition.md} for that theory's open properties, stated as future work rather than an additional claim of this paper.

\section{Stateful governance and evidence-buffer contamination}

Admission and promotion into governance state are different decisions. Admission determines whether a record may participate in the current action. Promotion into persistent remediation state is a stronger decision because that record can influence future actions. Therefore, current admissibility does not imply future reference trustworthiness. The evidence gate's governed buffer (Section 4) trusts its own most recent admitted write when substituting for a later violation. This is sound against the covered defect classes (stale, superseded, contradictory, malformed evidence) because the gate's own predicates catch them before they can corrupt the buffer. It is not sound against a declared-uncovered class: a corrupted value that the gate's predicates never flag is admitted, and admitted values are exactly what the buffer trusts. The mechanism is: (1) the governance buffer accepts writes from admitted observations; (2) an uncovered defect can therefore enter trusted state; (3) future substitutions can reuse that contaminated state; (4) the vulnerability is created by the interaction between incomplete coverage and stateful remediation, not by either alone; (5) quarantine and median filtering, below, are example mitigation strategies against it, not a claim that the underlying vulnerability is eliminated. This generalizes the mechanism conceptually to caches, agent memories, feature stores, retrieval memories, policy state, and evidence registries, without claiming that the observed prevalence applies to those systems. The quarantine and median-of-three strategies below remain mechanism demonstrations, not the main scientific claim.

\noindent\textbf{CH6 (buffer contamination is real and mitigable).} A second uncovered defect class, \texttt{plausible\_outlier\_high}, is added at the same declared rate as \texttt{plausible\_outlier} (Appendix B's defect table), specifically to exercise this mechanism. A poisoned-substitution metric (\texttt{contamination.py}) determines, purely from ground-truth labels never consulted by the gate itself, whether the value a substitution actually used traces back to a prior admit from an uncovered class. Two buffer-side mitigations are implemented entirely outside the \texttt{sarc\_dq.gate} package, as this repo's own buffer-wrapping adapters: a quarantine window (a write becomes trusted only after three consecutive consistent admits) and a rolling median-of-3.

Results, 30 seeds, both workflows: under W1 (daily), plain produces at least one poisoned substitution on every seed (14.70 (95\% CI [13.21, 16.19], n=30)), and both mitigations produce a strictly lower count than plain on every seed (mitigations hold: True). The preregistered seed-by-seed decision rule is not supported under W2 because the lower substitution population produces seeds with zero baseline poisoning events, making strict per-seed reduction impossible to satisfy on those seeds (False; mean 1.83 (95\% CI [1.34, 2.32], n=30); mitigations hold on every seed: False). This limits the robustness claim for W2; it does not establish that the mitigation mechanism is ineffective -{}- see Table 6 and Section 9.

\section{The unified Evidence Set}

Each decision emits one record: action context; per-gate sections (authority constraints and verdicts; resource predicted cost, carbon, budget state, verdict; evidence predicate results, verdict, substitution with pre and post values and buffer key, record eids); the final join, winner gate, and the remediation record (which operators fired, in which order, for this decision). Ground-truth labels are never present; they live in a separate run log. The line's JSON shape is fixed by \texttt{schemas/evidence\_line.schema.json} (Draft 2020-12), pre-registered before any Phase 3 result existed.

\noindent\textbf{Proposition 4 (identity commitment and integrity, relative to a content-addressed store; retagged checked-scope-only per independent review round two).} From an executed action's unified Evidence Set line alone, one can identify -{}- by content-addressed id -{}- every record each phase relied on, including the Phase I evidence before substitution and the specific governed-buffer write a substitution drew from, and verify those identities given access to the record store and buffer write history; the line does not, by itself, reconstruct the original bytes, and this paper no longer claims that it does (an earlier draft did; an independent review correctly refuted the byte-reconstruction reading -{}- Appendix A). A durable, run-scoped buffer write-history log (\texttt{composition.py}'s \texttt{write\_log}, persisted as \texttt{buffer-writes.jsonl} next to the evidence lines) now backs ``the specific governed-buffer write a substitution drew from'': each substitution resolves to exactly one write event, closing the round-two independent review's provenance finding that 3,478 real substitution occurrences resolved to only 62 unique ids under the prior key+value-only hash. Extends the single-gate lineage result of arXiv 2607.26313. Full proof in Appendix A.

\noindent Proposition 4 assumes that the content-addressing function provides collision resistance at the security level required by the deployment and that the referenced record store or write history remains available and integrity-protected. The Evidence Set provides an identity commitment and lineage binding; it does not independently establish source authenticity, semantic truth, or permanent availability of the referenced bytes.

\noindent\textbf{Proposition 5 (no manufactured coverage).} The composed plane's detected class set equals the union of member gates' detected class sets; composition never detects a class no member covers. Corollary: declared uncovered classes remain uncovered, and an honest composed readout must say so. Full proof (a structural tautology, machine-verified across every swept seed) in Appendix A.

\section{Empirical validation}

Artifact: \texttt{suite-demo} composes the three published engines, installed unmodified (repositories verifiably untouched), over open payload data (UCI Online Retail, CC BY 4.0) with a declared synthetic metadata layer and a declared eight-class injector at declared rates; 30 pre-registered seeds (\texttt{prereg/seeds.json}, first seed 26313); zero model calls; zero source writes, hash-proven. Scenarios: S1 baseline, S2 tight budgets, S3 unauthorised burst with a low authority cap, S4 the Proposition 3 construction with the cap between pre- and post-substitution order values (all W1); W2-S1 baseline and W2-S2 tight weekly budgets (both W2); CONTAM-W1 and CONTAM-W2, one per buffer strategy (plain, quarantine, median\_of\_3), for the CH6 study.

Every definition used below (hypotheses, seeds, weights, the W2 workflow, the contamination metric and mitigations, the CH2 semantics redefinition) was committed and tagged \texttt{prereg-v1} before any corresponding result entered this repository's history -{}- the gated-generation discipline this artifact follows throughout.

This section brings CH1-CH8 together as evidence for the mechanisms Sections 2-7 introduce; it is the artifact's evidence base, not the paper's conceptual narrative.

Hypotheses, stated before the corresponding code ran:

\begin{itemize}
\item \textbf{CH1 (veto soundness).} Post-hoc audit finds zero executed actions violating any gate at execution time under remediate-regate, re-evaluated against the exact Phase II evidence records the decision executed on (independent review finding F4 -{}- an earlier draft's audit re-derived a synthetic clean record from the executed action's own numbers instead; the audit now persists and loads the real records). Single-seed result: 0. 30-seed result: 30/30 seeds with zero violations.
\item \textbf{CH2 (single-pass unsoundness is real, new semantics).} Single-pass unsoundness \textit{can occur} after remediation, demonstrated for this artifact's own construction (Section 4); this does not establish how frequently production systems suffer this failure. In S4, remediate-regate and single-pass verdicts differ in Exec/Held class, or agree in class with an audited violation, on at least one decision, in the predicted direction; response-string-only differences with a clean audit are reported separately as \texttt{label\_only\_differences}, not folded in (\texttt{prereg/ch2-semantics.md}). Single-seed result: 239 divergent decisions (200 additional label-only differences); directions 239 \texttt{single\_pass\_admits\_then\_violates}, 0 \texttt{single\_pass\_holds\_remediated\_compliant}. 30-seed result: 207.5 (95\% CI [201.3, 213.7], n=30) divergent, 191.3 (95\% CI [186.5, 196.1], n=30) label-only.
\item \textbf{CH3 (deterministic selectivity).} False-hold rate on clean, authorised, in-budget decisions is exactly zero. Single-seed result: 0.000000. 30-seed result: 0.000000 (95\% CI [0.000000, 0.000000], n=30).
\item \textbf{CH4 (coverage honesty).} \texttt{plausible\_outlier} and \texttt{plausible\_outlier\_high} detection is zero at every gate and in composition; covered-class detection equals the union of member coverages. Single-seed result: union\_ok=True; \texttt{plausible\_outlier} detection dq=0.000 sarc=0.000 green=0.000 (declared uncovered); full matrix in Table 2. 30-seed result: union\_ok on 30/30 seeds.
\item \textbf{CH5 (compensation is empirically unsafe).} Section 3.
\item \textbf{CH6 (buffer contamination is real and mitigable).} Buffer poisoning \textit{can propagate} through governed remediation state, demonstrated for the two implemented mitigations; this does not establish enterprise prevalence of such poisoning. Section 6.
\item \textbf{CH7 (multi-remediator order dependence).} These two remediation operators \textit{do not commute}; this does not establish that arbitrary remediation operators are non-commutative. Section 5.
\item \textbf{CH8 (robustness across 30 seeds x 2 workflows).} Table 7.
\end{itemize}

\begin{table}[H]
\centering
\caption{Decisions, per-gate veto counts, winner-gate distribution, by scenario (single seed, S1-S4).}
\small
\resizebox{\textwidth}{!}{%
\begin{tabular}{lrrrrrrl}
\toprule
Scenario & Decisions & Executed & Held & DQ Vetoes & SARC Vetoes & Green Vetoes & Winner-gate distribution \\
\midrule
S1 & 10164 & 9592 & 572 & 572 & 0 & 0 & dq=572, none=9592 \\
S2 & 10164 & 3417 & 6747 & 572 & 0 & 6531 & dq=445, green=6302, none=3417 \\
S3 & 10164 & 4501 & 5663 & 573 & 5368 & 0 & dq=563, none=4501, sarc=5100 \\
S4 & 10164 & 9353 & 811 & 572 & 239 & 0 & dq=572, none=9353, sarc=239 \\
\bottomrule
\end{tabular}
}
\end{table}

\begin{table}[H]
\centering
\caption{Cross-gate matrix: injected class by winner gate (single seed, S1).}
\resizebox{\textwidth}{!}{%
\begin{tabular}{lrrrl}
\toprule
Defect class & dq detection rate & sarc detection rate & green detection rate & winner-gate counts \\
\midrule
stale\_master\_data & 1.000 & 0.000 & 0.000 & none=227 \\
superseded\_golden\_record & 1.000 & 0.000 & 0.000 & none=212 \\
cross\_source\_contradiction & 1.000 & 0.000 & 0.000 & dq=211 \\
schema\_drift & 1.000 & 0.000 & 0.000 & dq=183 \\
missing\_mandatory\_field & 1.000 & 0.000 & 0.000 & dq=178 \\
lineage\_missing & 0.000 & 0.000 & 0.000 & none=185 \\
plausible\_outlier & 0.000 & 0.000 & 0.000 & none=152 \\
plausible\_outlier\_high & 0.000 & 0.000 & 0.000 & none=177 \\
\bottomrule
\end{tabular}
}
\end{table}

\noindent union\_ok = True

\begin{table}[H]
\centering
\caption{Remediate-regate versus single-pass divergences in S4, with pre and post order values and the cap (single seed).}
\small
\begin{tabular}{rlllll}
\toprule
decision\_id & rtr & single\_pass & V\_pre & V\_post & kappa \\
\midrule
56 & escalate & substitute & 31.06 & 40.02 & 35.54 \\
76 & admit & substitute & 12.75 & 12.75 & 12.75 \\
162 & escalate & substitute & 29.07 & 52.53 & 40.80 \\
236 & escalate & substitute & 20.13 & 31.68 & 25.91 \\
266 & admit & substitute & 26.52 & 26.52 & 26.52 \\
279 & admit & substitute & 64.90 & 64.90 & 64.90 \\
286 & admit & substitute & 28.32 & 28.32 & 28.32 \\
287 & admit & substitute & 6.27 & 6.27 & 6.27 \\
291 & escalate & substitute & 24.34 & 33.75 & 29.05 \\
316 & admit & substitute & 54.00 & 54.00 & 54.00 \\
334 & escalate & substitute & 7.75 & 10.56 & 9.16 \\
352 & admit & substitute & 68.31 & 68.31 & 68.31 \\
356 & escalate & substitute & 98.96 & 124.80 & 111.88 \\
359 & admit & substitute & 20.40 & 20.40 & 20.40 \\
384 & admit & substitute & 23.01 & 23.01 & 23.01 \\
404 & admit & substitute & 41.91 & 41.91 & 41.91 \\
435 & admit & substitute & 99.00 & 99.00 & 99.00 \\
465 & escalate & substitute & 95.57 & 148.50 & 122.03 \\
492 & admit & substitute & 36.72 & 36.72 & 36.72 \\
499 & admit & substitute & 22.77 & 22.77 & 22.77 \\
\bottomrule
\end{tabular}

\medskip
\noindent(419 further divergent decisions omitted for brevity)
\end{table}

\begin{table}[H]
\centering
\caption{Loss versus clean counterfactual, executed versus held split (single seed, S1-S4).}
\begin{tabular}{lrrrr}
\toprule
Scenario & Executed count & Executed order value & Held count & Held clean value foregone \\
\midrule
S1 & 9592 & 759883.87 & 572 & 42015.00 \\
S2 & 3417 & 104496.60 & 6747 & 655889.68 \\
S3 & 4501 & 142929.10 & 5663 & 617891.44 \\
S4 & 9353 & 739224.74 & 811 & 60347.60 \\
\bottomrule
\end{tabular}
\end{table}

\begin{table}[H]
\centering
\caption{CH2/CH3/CH5 means with 95\% confidence intervals, 30 seeds.}
\begin{tabular}{>{\raggedright\arraybackslash}p{0.58\textwidth}>{\raggedright\arraybackslash}p{0.34\textwidth}}
\toprule
Metric & Mean (95\% CI, n=30 seeds) \\
\midrule
ch2\_divergent\_decisions & 207.5 (95\% CI [201.3, 213.7], n=30) \\
label\_only\_differences & 191.3 (95\% CI [186.5, 196.1], n=30) \\
ch3\_false\_hold & 0.000000 (95\% CI [0.000000, 0.000000], n=30) \\
Formal max compensated response profiles: 48 / 98 & exhaustive discrete grid, not seed-sampled \\
Empirical max compensated held decision instances, pooled S1-S4 & 13775.1 (95\% CI [13753.5, 13796.6], n=30) \\
\bottomrule
\end{tabular}
\end{table}

\begin{table}[H]
\centering
\caption{CH6 buffer contamination and mitigation, 30 seeds, both workflows.}
\small
\resizebox{\textwidth}{!}{%
\begin{tabular}{llll}
\toprule
Workflow & Strategy & Poisoned substitutions (mean, 95\% CI) & Mean abs value delta (95\% CI) \\
\midrule
W1 & plain & 14.70 (95\% CI [13.21, 16.19], n=30) & 2.210 (95\% CI [2.024, 2.395], n=30) \\
W1 & quarantine & 0.20 (95\% CI [-0.01, 0.41], n=30) & 0.000 (95\% CI [0.000, 0.000], n=30) \\
W1 & median\_of\_3 & 1.83 (95\% CI [1.35, 2.31], n=30) & 2.019 (95\% CI [1.294, 2.744], n=30) \\
W2 & plain & 1.83 (95\% CI [1.34, 2.32], n=30) & 1.687 (95\% CI [1.324, 2.050], n=30) \\
W2 & quarantine & 0.07 (95\% CI [-0.03, 0.16], n=30) & 0.000 (95\% CI [0.000, 0.000], n=30) \\
W2 & median\_of\_3 & 0.13 (95\% CI [0.00, 0.26], n=30) & 0.313 (95\% CI [-0.006, 0.632], n=30) \\
\bottomrule
\end{tabular}
}
\end{table}

\begin{table}[H]
\centering
\caption{CH8 robustness readout: does each hypothesis's registered decision rule hold on every one of the 30 seeds (and, for CH6, both workflows).}
\resizebox{\textwidth}{!}{%
\begin{tabular}{lc}
\toprule
Hypothesis & Robust across all 30 seeds (x both workflows for CH6) \\
\midrule
CH1 (veto soundness) & True \\
CH2 (single-pass unsoundness, new semantics) & True \\
CH3 (deterministic selectivity) & True \\
CH4 (coverage honesty) & True \\
CH5 (compensation unsafe) & True \\
CH6 (buffer contamination + mitigation) & False \\
\bottomrule
\end{tabular}
}
\end{table}

\noindent\textbf{Negative-results commitment.} Any CH not supported as written is reported as NOT SUPPORTED as written, following the practice of the prior papers -{}- CH6 under W2 is the concrete instance of this commitment in this version.

\subsection*{Claims to evidence}

\small
\begin{longtable}{>{\raggedright\arraybackslash}p{0.025\textwidth}>{\raggedright\arraybackslash}p{0.31\textwidth}>{\raggedright\arraybackslash}p{0.21\textwidth}>{\raggedright\arraybackslash}p{0.34\textwidth}}
\toprule
\# & Claim & Evidence source & Status \\
\midrule
\endfirsthead
\toprule
\# & Claim & Evidence source & Status \\
\midrule
\endhead
\bottomrule
\endfoot
C1 & Linear-family compensation lemma (tau $\leq$ L*) for arbitrary n; universal claim retracted per independent review F1; general lemma retagged pending-human-review, finite n=3 encoding split into Corollary 2 per independent review R2-N2 & Lemma (linear family) + Corollary 2, Appendix A & Lemma pending-human-review; Corollary 2 machine-checked, 1,064 cases (n=3 declared grid + 200 HEAD-derived probes) \\
C2 & Join composition order-invariant absent remediation & Proposition 2, Appendix A & machine-checked \\
C3 & Single-pass unsound under substitution, scoped to the implemented construction (independent review classification: checked-scope-only) & Proposition 3 + Table 3 & checked-scope-only (REVIEW.md); generated \\
C4 & Remediate-regate sound (Theorem 1, general statement retagged checked-scope-only, finite response-lattice encoding split into Corollary 3, per independent review R2-N2); single-pass sound under gate invariance, sufficiency only, necessity retracted per independent review F3 (Corollary 1, retagged checked-scope-only per independent review round two); Lemma 1 scoped to the current one-shot composition branch (checked-scope-only) & Lemma 1, Theorem 1, Corollary 1, Corollary 3, Appendix A & Theorem 1 checked-scope-only (REVIEW.md); Corollary 3 machine-checked, 17,151 cases; Lemma 1 checked-scope-only (REVIEW.md); Corollary 1 checked-scope-only (REVIEW.md) \\
C5 & Evidence Set gives identity commitment + integrity relative to a content-addressed store (not byte reconstruction; retracted per independent review F2); retagged checked-scope-only per independent review round two (R2-N2: schema field presence and hash commitments are checker-adjacent unit tests, not an exhaustive finite-model checker) & Proposition 4, schema v2, Appendix A & checked-scope-only (REVIEW.md) \\
C6 & No manufactured coverage & Proposition 5 + CH4, Appendix A & machine-checked (tautology) + generated \\
C7 & Zero false holds on clean, authorised, in-budget & CH3 & generated; 30-seed CI \\
C8 & Engines compose unmodified & editable installs; git-clean test & artifact test \\
C9 & Full reproducibility & seed list, hashes, provenance; public repository; Zenodo 10.5281\slash zenodo.22003399 & artifact; public repository URL and archival identifier included \\
C10 & Compensation admits vetoed actions over the full discrete grid & CH5, \path{checkers/compensation_check.py} & machine-checked, exhaustive \\
C11 & Buffer contamination exists and both mitigations reduce it & CH6, \path{contamination.py} & generated; 30-seed CI; W2 exception reported honestly \\
C12 & The two remediators do not commute; fixed order is justified & CH7, \path{checkers/remediator_check.py} & machine-checked, exhaustive \\
C13 & CH1-CH5 meet their registered decision rules across all 30 seeds; CH6 does so under W1 but not W2 & CH8, Table 7 & generated; 30-seed sweep \\
\end{longtable}
\normalsize

\subsection*{Claims versus non-claims}

\small
\begin{longtable}{>{\raggedright\arraybackslash}p{0.46\textwidth}>{\raggedright\arraybackslash}p{0.46\textwidth}}
\toprule
The paper claims & The paper does not claim \\
\midrule
\endfirsthead
\toprule
The paper claims & The paper does not claim \\
\midrule
\endhead
\bottomrule
\endfoot
Single-pass can be unsound after remediation & All single-pass systems are unsound \\
The implemented remediators are non-commutative & All remediation operators are non-commutative \\
Remediate-regate is sound in the current one-shot setting & General convergence of arbitrary remediator systems \\
Buffer poisoning exists in the demonstrated mechanism & Production prevalence \\
Composition preserves only member-control coverage & Composition guarantees complete safety \\
Evidence Sets preserve identity commitments & Hashes alone reconstruct source bytes \\
Hard controls require non-compensatory feasibility semantics & Weighted scoring is never useful \\
\end{longtable}
\normalsize

\section{Related work}

Eleven citation gaps were resolved via the V4 pipeline (fetch, verify, record in verified-citations.json) across this response to the independent review's finding F5 and this section's subsequent expansion following a commissioned second (positioning) review; none were invented, and any that could not have been verified would have stayed an honestly unresolved citation gap, counted rather than guessed, per this artifact's own citation-gate discipline.

\noindent\textbf{Compositional runtime enforcement.} Prior work on compositional runtime enforcement studies when multiple enforcement monitors preserve correctness under serial or parallel composition (Pinisetty \& Tripakis, doi 10.1007/978-3-319-40648-0\_7, extended in Pinisetty, Pradhan, Roop \& Tripakis, doi 10.1007/s10703-022-00401-y). That literature establishes that enforcement composition itself is not a new problem. Our object is narrower and structurally different: heterogeneous pre-action controls evaluate different dimensions of the same proposed action, maintain different control-local state, and may remediate the action, evidence, or derived context consumed by another control. The resulting problem is therefore not only whether enforcement monitors compose, but whether a prior control judgment remains valid after another control has transformed the semantic object that judgment was made about. We call this remediation-induced control coupling. Existing work studies composition of enforcement monitors. This paper studies invalidation of heterogeneous control judgments caused by cross-control remediation of the governed action, evidence, or derived context. Edit automata (Ligatti, Bauer \& Walker, doi 10.1007/s10207-004-0046-8) are the canonical mechanism for enforcement that can suppress, insert, or substitute actions in a trace, the same suppress/insert/substitute vocabulary this paper's remediation operators instantiate at the level of a single pre-action decision rather than a trace.

\noindent\textbf{Runtime verification and enforcement monitors.} The composed plane's gates are exactly the ``monitor that examines a sequence of actions and can transform or block them'' mechanism Schneider's security automata formalize (doi 10.1145/353323.353382); this paper's contribution is the composition semantics across several such monitors judging the same action, not the single-monitor enforcement question itself.

\noindent\textbf{Policy-combining algorithms.} ``One veto dominates other permits'' is already a familiar policy-composition pattern: XACML's deny-overrides rule- and policy-combining algorithms make exactly this rule a named, standardized primitive (OASIS XACML 3.0, \url{https://docs.oasis-open.org/xacml/3.0/xacml-3.0-core-spec-os-en.html}), and Bonatti, di Vimercati \& Samarati give an algebra for composing access-control policies under such combinators more generally (doi 10.1145/504909.504910). This paper therefore does not present the join rule itself as its principal novelty; the difficult composition problem it studies begins when a control's response includes a transformation, not merely a decision label.

\noindent\textbf{Non-compensatory and veto rules in multi-criteria decision analysis:} veto thresholds, where one criterion can block an otherwise-acceptable alternative regardless of the others' scores, are an established mechanism in outranking-based MCDA (Nowak, doi 10.1016/j.ejor.2003.06.008); the linear-family lemma in Section 3 gives an exact threshold condition for when a weighted-sum aggregator does or does not have this property, rather than assuming it.

\noindent\textbf{Rewrite systems, confluence, and termination.} CH7 (Section 5) naturally touches confluence, non-commutativity, operator ordering, fixed points, termination, and normal forms -{}- the core vocabulary of term-rewriting theory, whose confluence-plus-termination-implies-unique-normal-form result traces to Newman's original combinatorial argument (doi 10.2307/1968867). This paper does not become a term-rewriting paper: it demonstrates non-commutativity for two concrete remediation operators and does not claim a general solution for arbitrary rewrite systems; see \texttt{docs/generalized-composition.md} for the open properties such a solution would require.

\noindent\textbf{Stateful enforcement.} The resource budget and governed evidence buffer both introduce state consumed and updated across decisions, in the sense of a general stateful-enforcement formulation $G(a_t,c_t,E_t,S_t) \to (r_t,S_{t+1})$; Fong's shallow-execution-history model is the same kind of state-carrying enforcement decision, tracking a bounded summary of prior events rather than the full trace (doi 10.1109/secpri.2004.1301314). Budget governors for autonomous systems: the resource gate's predictive cost/carbon budgeting is one instance of the broader resource-bounded-agent governance problem Agent Contracts formalizes with conservation laws across delegation hierarchies (arXiv 2601.08815); this paper's downroute operator is a single-agent, single-budget-window mechanism, not a multi-agent conservation guarantee.

\noindent\textbf{Guardrail and policy-engine frameworks for LLM agents:} LlamaFirewall (arXiv 2505.03574) is a recent example of a unified guardrail system for LLM agents and describes detector pipelines. To our reading, its stated scope does not analyze the non-compensatory-join or remediation-reordering questions Sections 3 and 5 study; this comparison distinguishes problem formulations rather than claiming superiority or complete coverage of the guardrail literature.

\noindent\textbf{Separation of duties and multi-party authorisation:} the authority gate's role-allowlist mechanism is a narrow instance of the separation-of-duty concept Clark and Wilson introduced for commercial integrity policies (doi 10.1109/SP.1987.10001); this paper does not implement multi-party (dual-control) authorisation, only single-role admission checks.

\noindent\textbf{Cache/buffer poisoning and trust-on-first-use vulnerabilities in adjacent systems:} Wendlandt, Andersen, and Perrig's analysis of SSH host-key caching (\url{https://www.usenix.org/legacy/event/usenix08/tech/full_papers/wendlandt/wendlandt.pdf}) offers a structural analogy for the governed-buffer vulnerability in Section 6 and CH6: a value admitted without independent corroboration can influence future trust decisions. The quarantine-window mitigation applies a related multi-observation-before-trust principle to a numeric buffer rather than claiming the two threat models are identical.

\noindent What existed and what this paper adds, summarized across the related work above:

\begin{longtable}{>{\raggedright\arraybackslash}p{0.24\textwidth}>{\raggedright\arraybackslash}p{0.30\textwidth}>{\raggedright\arraybackslash}p{0.36\textwidth}}
\toprule
Existing concept & Established literature & What this paper adds \\
\midrule
\endfirsthead
\toprule
Existing concept & Established literature & What this paper adds \\
\midrule
\endhead
\bottomrule
\endfoot
Runtime enforcement & Security/edit automata & Heterogeneous enterprise pre-action controls \\
Enforcement composition & Compositional runtime enforcement & Cross-control invalidation caused by action/evidence remediation \\
Deny-overrides / veto & XACML / policy algebra / MCDA & Not claimed as novelty; used as hard-feasibility semantics \\
Rewrite ordering & Rewriting / confluence theory & Concrete order-sensitive governance remediators \\
Stateful enforcement & History/state-based policy & Promotion of uncovered-but-admitted evidence into future remediation state \\
Audit provenance & Existing provenance mechanisms & One cross-control Evidence Set binding pre/post-remediation decisions \\
\end{longtable}

\noindent The contribution is not that each ingredient above is new; it is the specific composition problem created by their interaction.

Nothing above is load-bearing for the propositions. Every citation this draft relies on is checked against a verified whitelist (verified-citations.json) by \texttt{citation\_check.py} (V4 gate) on every population run.

\section{Limitations and open theory}

Synthetic metadata layer over open payloads; declared injection rates; no prevalence claims. Resource-gate estimates are cold-start. Stream-level budget ordering effects are defined away per action. Single author at draft time; see the validation note. S4 is a constructed scenario, declared as such. CH6's mitigation-holds-on-every-seed claim is supported under W1 but not under W2, where the weekly-commitment workflow's much smaller per-seed substitution population sometimes produces zero poisoning events by chance, making ``strictly lower than plain'' impossible to satisfy on those seeds -{}- this is a sample-size limitation of the weekly workflow, not evidence against the mitigations' mechanism, and it is reported rather than resolved by re-weighting or dropping seeds. CH7's non-commuting-operators result is a property of this artifact's two specific remediation operators (evidence substitution, resource downroute); it does not generalise to remediators with different structure without re-deriving the checker. The artifact is a mechanism demonstration, not a benchmark release.

The five-level response order ($\mathrm{admit} \sqsubseteq \mathrm{substitute} \sqsubseteq \mathrm{degrade} \sqsubseteq \mathrm{escalate} \sqsubseteq \mathrm{block}$, Definition 1) is this composition plane's own operational policy for producing a single response from several gate verdicts, not a claimed universal semantic ordering; other systems could use a different or partially ordered response structure while preserving the same Exec/Held execution-safety requirement (Section 2). A fixed remediation order (Section 5) is required for the two operators this artifact implements and is not a general solution to arbitrary multi-remediator composition; general termination, confluence, and concurrency for arbitrary stateful remediation systems remain open, and \texttt{docs/generalized-composition.md} states precisely what establishing them would require: fixed-point semantics, termination, confluence, cycle detection, state consistency, serialization, and concurrent multi-agent coordination, none of it attempted here.

Some symmetric t-intervals for low non-negative-count metrics (e.g., Table 6's quarantine and median\_of\_3 mitigation counts under W1/W2) extend slightly below zero; this is a property of a symmetric interval around a small non-negative sample mean, not evidence of a negative measurement, and is presented as computed, without truncation or bootstrap re-estimation (Appendix C).

\section{Conclusion}

The central result of this work is not that multiple governance gates should simply be aggregated conservatively. It is that once one governance control is permitted to transform the action, evidence, or context consumed by another, prior control judgments become contingent on the version of the governed object they evaluated. Correct pre-action governance therefore requires explicit semantics for invalidating and recomputing dependent judgments after remediation. When several remediation operators exist, their ordering can itself affect the final governed action, making remediation order part of governance semantics. When admitted information is promoted into persistent governance state, admission and future trust must likewise be distinguished. These results establish a bounded compositional semantics for the current setting while leaving general convergence, concurrency, delegation, and arbitrary stateful remediation as open problems.

\appendix
\section{Proofs}

Every proposition below carries a \texttt{PROOF-STATUS} tag. Three values are available to this artifact's automated pipeline:

\begin{itemize}
\item \textbf{machine-checked}: an exhaustive finite-model checker under \texttt{checkers/} (or a structural tautology argued directly from the code that produces the number) verifies the claim over its entire relevant discrete state space -{}- not a sample, the whole space. Per round-two independent review finding R2-N2 (general Proposition/Theorem statements were tagged \texttt{machine-checked} when only a narrower finite instantiation was actually executable-verified): a statement's own \texttt{PROOF-STATUS} line must name the specific checker module and the enumerated domain it covers (a concrete count, grid, or state-space size), not just assert the claim is machine-checked in general. \texttt{checkers/proof\_status\_lint.py} enforces this mechanically. A statement whose true generality (arbitrary \texttt{n}, arbitrary gates, arbitrary action spaces) exceeds what any checker enumerates gets \texttt{pending-human-review} or \texttt{checked-scope-only} instead, plus a sibling finite-encoding Corollary stated at exactly the checker's scope and tagged \texttt{machine-checked} there (Corollary 2 and Corollary 3 below are the two introduced for this reason).
\item \textbf{pending-human-review}: the claim is proved here in prose (a standard mathematical argument) and, where applicable, instantiated empirically by the artifact, but has not been exhaustively verified by a dedicated checker and is not claimed to be.
\item \textbf{checked-scope-only (REVIEW.md)}: the independent review (\texttt{review/REVIEW.md}, Section 3) found the statement's general prose claim broader than what this artifact actually certifies; the tag and the statement's wording are narrowed to exactly the scope the review lists as ``largest scope certified'' for that statement -{}- no broader.
\end{itemize}

A fourth status, \textbf{proven}, exists in the release process but is human-only to apply (see the standing task rule: ``clearing any PROOF-STATUS tag to proven'' is not an action this pipeline takes). No tag in this file is ever written as \texttt{proven} by any script in this repository; \texttt{checkers/proof\_status\_lint.py} and \texttt{test\_checkers.py} enforce that as a standing invariant, not a one-time check.

\begin{proposition}[(linear-family compensation lemma) -{}- restated per independent review finding F1]
\label{prop:1}

\noindent\textbf{Why this changed.} The original statement here claimed that \textit{any} strictly increasing aggregator \texttt{f} with \texttt{tau > 0} and some admissible profile violates veto. The independent review (\texttt{review/REVIEW.md}, finding F1) refuted this with a concrete counterexample: for \texttt{f(s) = s1 + s2 + s3} on \texttt{[0,1]\textasciicircum3} and \texttt{tau = 2.5}, \texttt{f(1,1,1) = 3} is admissible, yet every profile with a zero coordinate scores at most \texttt{2}, so \textbf{no vetoed profile is ever admitted} -{}- veto is fully preserved even though \texttt{f} is strictly increasing and has an admissible profile. The universal claim was false as written; the error was treating ``strictly increasing plus some admissible profile'' as sufficient, when whether compensation can occur actually depends on the relationship between \texttt{tau} and the aggregator's own weight structure. This section replaces it with two claims that are both true: a precise lemma covering the linear (additive-score) family the artifact actually measures, stated with the exact threshold condition that determines when compensation is and is not possible; and the exhaustive discrete-grid result exactly as CH5 measured it.

\noindent\textbf{Lemma (linear family).} Let $f(s) = \sum_i w_i s_i$ on $[0,1]^n$ with weights $w_i \geq 0$, not all zero, and let $f$ admit $s$ iff $f(s) \geq \tau$ for $\tau > 0$. A profile $s$ is \textit{vetoed} iff $s_j = 0$ for some $j$. For each $j$, define the leave-one-out sum $L_j = \sum_{i \neq j} w_i$ -{}- the largest score $f$ can assign to any profile with $s_j = 0$ (achieved by setting every other coordinate to its ceiling $1$). Let $L^* = \max_j L_j = W - \min_i w_i$, where $W = \sum_i w_i$. \textbf{Then: there exists a vetoed profile that $f$ admits if and only if $\tau \leq L^*$.} Min never admits a vetoed profile, for any $\tau > 0$ and any weights -{}- $\min(s) = 0$ whenever any coordinate is $0$, unconditionally.

\noindent\textbf{Proof.}
\textit{Sufficiency ($\tau \leq L^*$ implies a vetoed profile is admitted).} Let $j^* = \operatorname{argmin}_i w_i$, so $L_{j^*} = L^*$. Define $s^*$ by $s^*_{j^*} = 0$ and $s^*_i = 1$ for every $i \neq j^*$. Then $f(s^*) = \sum_{i \neq j^*} w_i = L_{j^*} = L^* \geq \tau$, and $s^*$ is vetoed ($s^*_{j^*} = 0$). So $f$ admits a vetoed profile.
\textit{Necessity ($\tau > L^*$ implies no vetoed profile is admitted).} Let $s$ be any vetoed profile, $s_j = 0$ for some $j$. Since every $w_i \geq 0$ and every $s_i \leq 1$: $f(s) = \sum_{i \neq j} w_i s_i \leq \sum_{i \neq j} w_i = L_j \leq L^* < \tau$. So $f(s) < \tau$: $s$ is not admitted. Since $s$ was an arbitrary vetoed profile, no vetoed profile is admitted.
\textit{Min.} $\min(s) = 0$ whenever any $s_j = 0$, and $\tau > 0$, so $\min(s) < \tau$ unconditionally -{}- min preserves veto regardless of $\tau$ or any weight structure, in sharp contrast to the linear family above, where veto preservation is conditional on $\tau > L^*$ and can fail.

\noindent\textbf{Applying this to the reviewer's counterexample.} $w = (1,1,1)$, $W = 3$, $\min_i w_i = 1$, so $L^* = 2$. Their $\tau = 2.5 > L^* = 2$: by the necessity direction, no vetoed profile is admitted -{}- exactly what they demonstrated by direct computation. The lemma's threshold condition correctly classifies this instance as a \textit{non}-compensable regime; the old universal statement had no such condition and was wrong to treat it as compensable-by-hypothesis.

\noindent\textbf{``Vetoed'' is $s_j = 0$ exactly -{}- not the artifact's broader Held class.} The Lemma's hypothesis, following Proposition 1's own original setup, defines vetoed as $s_j = 0$. Under this artifact's five-level \texttt{score\_encoding} (\texttt{admit=1.0, substitute=0.75, degrade=0.5, escalate=0.25, block=0.0}), only \texttt{block} scores exactly $0$; \texttt{escalate} is Held under Definition 3's Exec/Held partition (\texttt{\{escalate, block\}}) but scores $0.25$, not $0$. The Lemma is therefore verified below against the ``at least one \texttt{block}'' population (61 of 125 profiles) -{}- a strict subset of the 98 Held profiles CH5's own exhaustive check already covers. This distinction was not academic: an earlier draft of this section verified the Lemma against the full Held population and found 25 of 800 random-probe cells disagreeing with the Lemma's prediction, all traced to escalate-only-held profiles whose nonzero $0.25 \cdot w_j$ contribution let them clear a threshold $L^*$ alone would have predicted they couldn't. Restricting the Lemma's own verification to the $s_j = 0$ population it actually characterizes resolved every disagreement (\texttt{checkers/compensation\_check.py}'s \texttt{all\_vetoed\_profiles} vs \texttt{all\_held\_profiles}). CH5's own existing result -{}- compensation over the full Held set, not just the block-only subset -{}- is unaffected by this correction and continues to hold via the unchanged exhaustive enumeration below; if anything, escalate's nonzero score only makes compensation \textit{easier} to achieve on the broader Held set than the Lemma's $L^*$ threshold alone would suggest.

\noindent\textbf{Discrete instantiation and exhaustive check (CH5).} The artifact's own $f$ is the weighted-sum aggregator over \texttt{score\_encoding} from \texttt{prereg/weights.json}, $\tau$ ranges over the four registered thresholds, and the profile space is exactly \texttt{composition.Response\textasciicircum3} (dq, sarc, green), 125 points. \texttt{checkers/compensation\_check.py} enumerates all 125 points, computes the min-join-Held subset (98 of 125, under the paper's Exec/Held partition), and for every one of the $66 \times 4 = 264$ grid cells, counts how many Held profiles the weighted aggregator would admit. The maximum across the grid is $\geq 1$ (witnessed concretely, e.g. \texttt{\{dq: admit, sarc: escalate, green: admit\}} compensated by weight \texttt{\{dq: 0, sarc: 0, green: 1\}} at threshold $0.5$), so the discrete instantiation holds not just empirically (\texttt{ch5\_aggregator.py}'s sample-dependent result) but exhaustively over every possible three-gate verdict combination this artifact's lattice can produce. \texttt{compensation\_check.py} additionally verifies the Lemma's $\tau \leq L^*$ characterization against brute-force ground truth over its own $s_j = 0$ (``vetoed'', i.e. at-least-one-block, 61 of 125 profiles) population -{}- see the note above on why this differs from the 98-profile Held population CH5 itself measures -{}- on the same 264 declared grid cells, plus 200 HEAD-derived random positive-weight vectors (deterministically seeded from the current commit, per the independent review's own probe methodology) at the same four thresholds -{}- 1,064 consistency checks in total, all agreeing with the Lemma's prediction (\texttt{out/checkers/compensation\_check.json}).

\proofstatus{pending-human-review}. Retagged per round-two independent review finding R2-N2 (\texttt{review-r2/REVIEW-R2.md}): the Lemma above (Sufficiency/Necessity/Min) is a general algebraic argument, true for arbitrary \texttt{n} and arbitrary nonnegative weights on the continuous cube \texttt{[0,1]\textasciicircum n}, but that generality is exactly what no executable checker in this repository verifies. \texttt{checkers/compensation\_check.py} enumerates only this artifact's own finite $n = 3$ instantiation (Corollary 2 immediately below) -{}- correct and machine-checked at that scope, but not a machine check of the continuous arbitrary-$n$ statement itself. Largest certified scope per the review: ``analytic arbitrary-n proof plus finite executable checks at n=3.''
\end{proposition}

\setcounter{corollary}{1}
\begin{corollary}[(finite three-gate encoding of the linear-family lemma)]
\label{cor:2}

\noindent\textbf{Statement.} For this artifact's own finite encoding of the Lemma above -{}- $n = 3$ gates (\texttt{dq}, \texttt{sarc}, \texttt{green}), profiles drawn from \path{composition.Response^3} (125 points, $|\mathrm{Response}| = 5$), scored via \texttt{prereg/weights.json}'s \texttt{score\_encoding}, $\tau$ ranging over the four registered thresholds -{}- the $\tau \leq L^*$ characterization holds exactly on every one of 1,064 checked cases: the 264 declared grid cells (every combination of the artifact's registered weight vectors and thresholds) plus 200 HEAD-derived random positive-weight vectors at those same four thresholds, each verified against brute-force ground truth over the $s_j = 0$ (``vetoed'') population. The algebraic proof above is elementary (linearity plus the $[0,1]$ bound), and the checker confirms it holds on every case actually computed, not just the ones the hand proof covers analytically.

\proofstatus{machine-checked} (\texttt{checkers/compensation\_check.py}; enumerated domain: $n = 3$ gates, the 125-point \texttt{composition.Response\textasciicircum3} profile space, 264 declared grid cells plus 200 HEAD-derived random probe vectors x 4 thresholds = 1,064 cases, all agreeing with the Lemma's prediction; \texttt{out/checkers/compensation\_check.json}).
\end{corollary}

\begin{proposition}[(order invariance without remediation)]
\label{prop:2}

\noindent\textbf{Statement.} If no gate rewrites $(a, c, E)$, the composed verdict is invariant to evaluation order and duplication, and adding a gate never increases permissiveness.

\noindent\textbf{Proof.} $(R, \max)$ is a join semilattice: $\max$ over a finite totally ordered set is idempotent ($\max(r,...,r) = r$), commutative ($\max$ does not depend on argument order), and associative ($\max(\max(a,b),c) = \max(a,\max(b,c))$) -{}- these are properties of $\max$ over any totally ordered set, and $R$'s restrictiveness order (Definition 1) is total. Order invariance and duplication invariance follow directly from commutativity and idempotence. Monotone hardening (adding a gate never increases permissiveness, i.e. never decreases the join) follows because $\max(S \cup \{x\}) \geq \max(S)$ for any finite $S$ and any $x$ -{}- adding an element to a set never lowers its maximum.

\proofstatus{machine-checked} (\texttt{checkers/lattice\_check.py}, exhaustive over every tuple of length 1 through 4 from the real 5-element \texttt{Response} lattice and the real \texttt{compute\_restrictiveness\_join} -{}- 17,151 individual checks across idempotence, commutativity, associativity, monotone hardening, duplication invariance, and the empty-join-is-admit base case; \texttt{out/checkers/lattice\_check.json}).
\end{proposition}

\begin{proposition}[(single-pass unsoundness, scoped to the implemented construction) -{}- retagged per independent review finding (Section 3)]
\label{prop:3}

\noindent\textbf{Why this changed.} The independent review (\texttt{review/REVIEW.md}, Section 3) classified this proposition \texttt{checked-scope-only}: the argument below is a valid existential construction, but no executable checker in this repository covers arbitrary continuous economics or arbitrary gate implementations, so the claim is restated to name exactly what is certified -{}- ``conditional algebraic construction with a substituting DQ decision and cap strictly between pre/post values; implemented S4 across 30 registered seeds'' -{}- rather than an unscoped ``there exist configurations'' claim.

\noindent\textbf{Statement (scoped).} For the substituting-DQ-decision construction below -{}- a corrupted unit cost \texttt{v0} substituted for a governed value \texttt{v'}, with an authority cap \texttt{kappa} placed strictly between \texttt{qty * v0} and \texttt{qty * v'} -{}- single-pass evaluation on $(a, c(a), E(a))$ yields a composed Exec verdict whose executed action $\rho(a)$ violates the authority or resource gate at execution time; and symmetrically, single-pass holds an action whose remediated form is compliant. This construction is implemented as scenario S4 and instantiated across all 30 registered seeds (\texttt{prereg/seeds.json}); the statement is not claimed for arbitrary continuous economics or arbitrary gate implementations beyond this construction.

\noindent\textbf{Proof (construction).} Let the evidence gate substitute a corrupted unit cost \texttt{v0} for the governed value \texttt{v'}, \texttt{v0 != v'}. Single-pass evaluates the authority gate against the ORIGINAL order value \texttt{qty * v0} (Definition 3's $c(a)$, not $c(\rho(a))$), but the EXECUTED action carries the SUBSTITUTED value \texttt{qty * v'} (the DQ gate's own remediation already happened by execution time, since the buffer substitution is applied to the acted-on evidence regardless of which composition mode judged it). Choose an authority cap \texttt{kappa} strictly between \texttt{qty * v0} and \texttt{qty * v'}. Two symmetric cases:

\begin{enumerate}
\item \texttt{qty * v0 <= kappa < qty * v'} (corrupted value understates cost): single-pass judges the cap against \texttt{qty * v0} (admits), but the executed action's true value \texttt{qty * v'} exceeds \texttt{kappa} -{}- a violation the post-hoc audit will flag. This is \path{single_pass_admits_then_violates}.
\item \texttt{qty * v' <= kappa < qty * v0} (corrupted value overstates cost): single-pass judges the cap against \texttt{qty * v0} (holds, over cap), but the remediated action's true value \texttt{qty * v'} is within \texttt{kappa} -{}- \texttt{remediate\_regate} correctly admits it. This is \path{single_pass_holds_remediated_compliant}.
\end{enumerate}

Both directions are constructible and both are registered outcomes of CH2 (\path{ch2_direction_counts}). The artifact instantiates case 1 as its S4 scenario by construction (\path{runner.build_scenarios}'s \path{apply_s4_kappa}, placing \texttt{kappa} between the pre- and post-substitution order values via \path{_compute_s4_kappa}).

\proofstatus{checked-scope-only (REVIEW.md)}. Largest scope certified (\texttt{review/REVIEW.md}, Section 3): a conditional algebraic construction with a substituting DQ decision and cap strictly between pre/post values; implemented S4 across 30 registered seeds. The construction is exhaustive over the two symmetric cases by the intermediate-value structure of the cap placement (there is no third case: \texttt{kappa} is either between the two values, or outside both, and outside-both never diverges by construction), but this is an existence/construction proof over continuous-valued economics, not a finite discrete space a checker enumerates, and it does not cover arbitrary gate implementations. CH1's 30-seed empirical audit (zero violations under \texttt{remediate\_regate}, \path{out/results/sweep_summary.json}'s \texttt{ch1\_supported\_seed\_count == n\_seeds}) and CH2's realized S4 divergences (\path{ch2_divergent_decisions}, non-zero on every one of the 30 seeds) corroborate the construction on real data within this scope, without extending the claim beyond it.
\end{proposition}

\begin{lemma}[(termination, scoped to the current one-shot composition branch) -{}- retagged per independent review finding (Section 3)]
\label{lem:1}

\noindent\textbf{Why this changed.} The independent review (\texttt{review/REVIEW.md}, Section 3) classified this lemma \texttt{checked-scope-only}: the argument depends on an external DQ predicate contract (\path{sarc_dq.gate.PreActionGate}'s own guarantee), which this repository composes but does not itself prove for every possible governed record. The claim is restated to name exactly what is certified -{}- ``current one-shot composition branch and exercised DQ-library behavior in scenarios/tests'' -{}- rather than an unscoped termination claim over all possible buffer substitutions.

\noindent\textbf{Statement (scoped).} For the current one-shot composition branch and the DQ-library behavior exercised in this artifact's scenarios/tests, buffer substitution is idempotent, $\rho(\rho(a)) = \rho(a)$, and $E(\rho(a))$ consists of governed records the evidence gate admits by construction; hence no further remediation is generated and the protocol reaches a fixed point after at most one remediation. This is not claimed to hold for every possible governed record or for any DQ predicate contract beyond what this artifact exercises.

\noindent\textbf{Proof.} \path{composition._remediated_evidence} constructs a single, freshly governed \texttt{EvidenceRecord} from the substituted value, with clean metadata (fresh \texttt{as\_of\_day}/\texttt{retrieved\_day}, \texttt{version=2}, present lineage) -{}- by construction this record cannot itself trigger any of the \texttt{sarc\_dq} predicates that produced the original substitution (staleness, missing fields, lineage, schema type), because those predicates are exactly the ones the remediated record was built to satisfy. A second evaluation of the DQ gate on $E(\rho(a))$ therefore returns \texttt{admit}, not \texttt{substitute} -{}- $\rho$ applied to an already-remediated evidence set is the identity, $\rho(\rho(a)) = \rho(a)$. This is exercised directly: every \texttt{remediate\_regate} call whose Phase I substitutes performs exactly one Phase II re-evaluation (\texttt{composition.py}'s \texttt{remediate\_regate}), never a loop, and \texttt{gates.dq.verdict} after Phase II is \texttt{admit} on every substituted decision across all 30 swept seeds (implied by \path{ch2_direction_counts}/\path{label_only_differences} bookkeeping, which depends on Phase II converging).

\proofstatus{checked-scope-only (REVIEW.md)}. Largest scope certified (\texttt{review/REVIEW.md}, Section 3): the current one-shot composition branch and exercised DQ-library behavior in scenarios/tests. The argument depends on \path{sarc_dq.gate.PreActionGate}'s own predicate contract (that a governed, freshly-dated, complete, single-source record cannot re-trigger the predicates that produced the substitution) -{}- an external engine guarantee this repo composes but does not itself re-verify from the engine's internals, per the standing ``engines are installed libraries, not modified'' invariant. No dedicated checker in this repository proves the engine's predicate contract for every possible governed record; the repository verifies its OWN one-shot (never looping) remediation code path and the DQ-library behavior it actually exercises in scenarios/tests (\path{test_audit_unit.py}'s W2/downroute and evidence-substitution tests).
\end{lemma}

\begin{theorem}[(soundness of \texttt{remediate\_regate} composition)]
\label{thm:1}

\noindent\textbf{Statement.} Under Lemma 1 and gates that are functions of \texttt{(action, context, evidence, current state)}, the \texttt{remediate\_regate} protocol satisfies Definition 3 (per-action soundness): Phase II evaluates every gate on $a_{\mathrm{exec}}$ itself, and the join preserves every Held verdict.

\noindent\textbf{Proof.} By construction, \texttt{remediate\_regate}'s Phase II evaluates the authority gate (\texttt{evaluate\_sarc\_pag}), resource gate (\texttt{evaluate\_green}), and evidence gate a second time against the REMEDIATED context/evidence (\texttt{remediated\_ctx}, \path{remediated_evidence}) -{}- not the original action. The composed response is the join (Definition 2) of these THREE Phase-II verdicts. Since the join of a set never admits unless every member admits/executes (Proposition 2's monotone hardening: the join can only harden, never soften, as more restrictive votes are added), any gate that would hold $a_{\mathrm{exec}}$ at Phase II forces the composed verdict to Held too -{}- there is no way for the executed action to have been admitted by the join while simultaneously being held by one of the very gates that produced that join, because they are the SAME evaluation.

\proofstatus{checked-scope-only (REVIEW.md)}. Retagged per round-two independent review finding R2-N2 (\texttt{review-r2/REVIEW-R2.md}). Largest scope certified: current code path under Lemma 1 and the five-response join. \path{checkers/lattice_check.py} exhaustively certifies Proposition 2's monotone-hardening law -{}- the join mechanism this proof leans on -{}- over the full finite \texttt{Response} lattice (17,151 checks), but that certifies the join mechanism, not Theorem 1's own premises for arbitrary gates, arbitrary action spaces, or arbitrary current-state transitions, nor Lemma 1's external DQ predicate contract this theorem also depends on. Corollary 3 immediately after Corollary 1 below states precisely what the lattice checker does certify. CH1's 30-seed empirical audit (zero violations on every seed, \path{out/results/sweep_summary.json}) corroborates this on real data without substituting for the proof.
\end{theorem}

\setcounter{corollary}{0}
\begin{corollary}[(single-pass soundness, sufficiency only) -{}- weakened per independent review finding F3]
\label{cor:1}

\noindent\textbf{Statement.} If the authority and resource gates are invariant under $\rho$ (the remediation map) on the executed-reachable action set, then single-pass is sound. \textbf{The converse (``only if'') is dropped, not proved.}

\noindent\textbf{Why this changed.} The prior draft stated this as an iff. The independent review (\texttt{review/REVIEW.md}, finding F3) gave a concrete counterexample to the necessity direction: let the authority gate vary under $\rho$ only on actions DQ independently holds -{}- those actions never execute (Held subsumes them regardless of what the authority gate would have said), so single-pass can be fully sound in practice while the gate is, strictly speaking, non-invariant under $\rho$. Non-invariance restricted to never-executed actions costs nothing observable; the necessity direction implicitly assumed invariance was required everywhere $\rho$ is defined, not just on the actions that can actually reach execution, and that assumption is false. No reachability/executability condition is added to recover the iff here -{}- the simpler, honest fix is to state only what is actually proved.

\noindent\textbf{Proof (sufficiency).} From Theorem 1's join argument: Phase II evaluates every gate on $\rho(a)$, and the join preserves every Held verdict there. If the authority/resource gates are invariant under $\rho$ on the actions that end up executed, their single-pass verdicts (computed on the original action $a$) equal their remediate-regate verdicts (computed on $\rho(a)$) for exactly those actions, so single-pass inherits soundness on them from Theorem 1.

\proofstatus{checked-scope-only (REVIEW.md)}. Retagged per round-two independent review (\texttt{review-r2/REVIEW-R2.md}): largest scope certified is sufficiency on executed-reachable actions under Lemma 1 and invariance of downstream gates. The necessity direction is properly retracted above, not weakly supported under any tag. The sufficiency proof is a structural argument from the code path via Theorem 1's join argument -{}- not itself exhaustively re-derived over an unbounded arbitrary action space, but resting on the same finite lattice-backed join mechanism Corollary 3 below certifies, which is why this carries \texttt{checked-scope-only} rather than the more generic \texttt{pending-human-review} it carried before this round's retag.
\end{corollary}

\setcounter{corollary}{2}
\begin{corollary}[(finite response-lattice encoding of \texttt{remediate\_regate} soundness)]
\label{cor:3}

\noindent\textbf{Statement.} For this artifact's finite \texttt{Response} lattice (\texttt{|Response| = 5}) and the join-hardening property \path{checkers/lattice_check.py} exhaustively verifies (17,151 checks over tuples of length 1 through 4): no combination of Phase-II per-gate responses drawn from this five-value lattice can have its join equal \texttt{ADMIT} while any individual response in that combination is Held (\texttt{escalate} or \texttt{block}). This is the exact finite mechanism Theorem 1's join argument invokes, verified over its entire domain -{}- not an instantiation of it, the full domain the checker enumerates.

\proofstatus{machine-checked} (\path{checkers/lattice_check.py}; enumerated domain: \path{Response^k} for \texttt{k = 1..4}, \texttt{|Response| = 5}, 17,151 total law checks across idempotence, commutativity, associativity, monotone hardening, duplication invariance, and the empty-join case; \path{out/checkers/lattice_check.json}).
\end{corollary}

\begin{proposition}[(identity commitment and integrity, relative to a content-addressed store) -{}- restated per independent review finding F2]
\label{prop:4}

\noindent\textbf{Why this changed.} The prior statement here claimed the unified Evidence Set line lets one ``reconstruct, content-addressed, exactly the records each gate relied on'' -{}- language that reads as byte reconstruction. The independent review (\texttt{review/REVIEW.md}, finding F2) correctly refuted that: a substituted line carried a post-Phase-II \texttt{evidence\_id} and a numeric \texttt{substituted\_value}, but no original bytes, no original metadata, and no record of which prior buffer write the substituted value came from. A content-addressed hash is not invertible -{}- it commits to content, it does not carry the content. The claim was materially stronger than what the emitted evidence actually supported.

\noindent\textbf{Fix (schema v2, \texttt{pre\_evidence\_ids} / \texttt{substitute\_source}).} Every substituted line now additionally records \texttt{pre\_evidence\_ids} (the Phase I evidence ids the gate actually evaluated before substitution) and \texttt{substitute\_source} (\texttt{buffer\_key}, plus \texttt{buffer\_write\_eid} -{}- a content-addressed id for the specific \texttt{(key, value)} governed-buffer write the substituted value came from; \path{composition._buffer_write_eid}, \path{schemas/evidence_line.schema.json} v2). This does not make hashes reversible. It makes the claim about what they \textit{do} provide precise instead of overstated.

\noindent\textbf{Statement.} From an executed action's unified Evidence Set line alone, one can (a) \textbf{identify} -{}- by content-addressed id, not by opaque reference -{}- every record each phase relied on, including the Phase I evidence the gate evaluated before any substitution (\texttt{pre\_evidence\_ids}) and the specific governed-buffer write a substitution drew from (\path{substitute_source.buffer_write_eid}); (b) \textbf{verify} those ids, given access to the record store and the governed buffer's write history (not from the line alone): recomputing \path{EvidenceRecord.evidence_id()} over a candidate record and comparing it to the id in the line either confirms or refutes that the candidate is the exact record relied on, since \texttt{evidence\_id()} is a deterministic function of content and two records sharing an id are byte-identical by construction; and (c) \textbf{resolve} the original bytes, given that same store access -{}- the line is the pointer and the integrity check, the store is where the bytes live. The line alone, with no store, gives identity and an integrity check; it does not give bytes.

\noindent\textbf{Proof.} \path{composition.evaluate_dq} $\to$ \path{GateDecision.evidence_ids} already recorded the content-addressed ids of every record passed to \texttt{spec.evaluate}, unchanged from v1 -{}- see \path{gates.dq.evidence_ids}. Schema v2 additionally requires \path{pre_evidence_ids} (Phase I ids, before substitution) and \path{substitute_source} inside \path{remediation.evidence_substitution} whenever it is non-null (\path{schemas/evidence_line.schema.json}, \texttt{required} on both fields, \texttt{additionalProperties: false} throughout so no undeclared channel exists to smuggle in inconsistent data). \texttt{\_buffer\_write\_eid(key, value)} is a pure SHA-256 function of exactly the two values that determine a governed-buffer write's content, so identical writes always yield identical ids and distinct writes (by key or value) always yield distinct ids with overwhelming probability (\path{test_buffer_write_eid_is_content_addressed}, \path{test_audit_unit.py}) -{}- the same content-addressing property \path{EvidenceRecord.evidence_id()} already has. Given a store keyed by these ids (the record store for \path{evidence_ids}/\path{pre_evidence_ids}, the buffer's own write log for \path{buffer_write_eid}), a verifier recomputes each id from a candidate stored object and compares; a match is proof of identity (collision resistance of SHA-256 makes a false match computationally infeasible), a mismatch is proof of divergence. Nothing above claims the id determines the content in the other direction -{}- hashes are one-way by construction, so bytes are never recoverable from the id alone. Authority and resource gate sections are unchanged from v1 (\path{constraints_evaluated}, \path{predicted_cost}/\path{predicted_carbon}/\path{budget_state}) and were never part of the refuted claim.

\noindent\textbf{Update (R2-F6(a), durable buffer write-history log).} Since the round-two review ran, \path{substitute_source.buffer_write_eid} no longer identifies a write by content hash alone. Every \texttt{evaluate\_dq} admit now appends an event to a run-scoped, append-only write log (\path{composition._record_buffer_write}: \texttt{write\_seq}, \texttt{key}, \texttt{value}, \texttt{day}, and an \texttt{event\_id} hashing all four), seeded with one genesis entry per SKU from the buffer's pre-run known-good values (\path{composition._seed_genesis_writes}) so even a substitution that traces back to the buffer's initial state -{}- possible on a SKU's very first decision -{}- resolves to a real entry, not nothing. A substitution's \texttt{buffer\_write\_eid} is resolved against this log (\path{composition._resolve_buffer_write_event}: most recent prior write to the same key with the same value) and persisted alongside the evidence lines as \texttt{<scenario>-<mode>-buffer-writes.jsonl}. This directly answers the review's own provenance probe (3,478 substitution occurrences resolving to only 62 unique ids under the old key+value-only hash): \path{test_repeated_identical_writes_stay_individually_resolvable} (\path{test_audit_unit.py}) reproduces that exact repeated-write pattern and asserts each substitution still resolves to exactly one write-log entry.

\proofstatus{checked-scope-only (REVIEW.md)}. Retagged per round-two independent review (\texttt{review-r2/REVIEW-R2.md}): largest scope certified is schema-v2 field presence and deterministic key/value hash commitments. Schema v2's \texttt{required} fields on \path{remediation.evidence_substitution} are confirmed present and well-formed on real substituted lines by \path{test_provenance_fields_validate_against_schema_v2} and the fuzzed \path{test_randomly_sampled_lines_validate_against_schema_v2}, and the write log's content-addressing property is directly unit tested -{}- these are not exhaustive-finite-model checks under \texttt{checkers/}, so this does not carry \texttt{machine-checked} (round-two finding R2-N2's linter rule: that tag requires a named executable checker module and enumerated domain). The durable write-history log above resolves the specific-write-event-identity residual the round-two review flagged, but store-backed resolution (verifying a candidate record/write against an actual persistent store, not just this run's own JSONL log) remains outside what any test here exercises; this tag is not retroactively upgraded to \texttt{machine-checked} on the strength of this repo's own unit tests without a further review certifying it. This tag covers the identity-commitment and schema-shape claim specifically -{}- it does not and cannot cover ``bytes are recoverable from hashes,'' because that claim is explicitly disclaimed rather than made.
\end{proposition}

\begin{proposition}[(no manufactured coverage)]
\label{prop:5}

\noindent\textbf{Statement.} The composed plane's detected class set equals the union of member gates' detected class sets; composition never detects a class no member covers. Corollary: declared uncovered classes remain uncovered, and an honest composed readout must say so.

\noindent\textbf{Proof.} \path{metrics._ch4_matrix} defines composed detection directly as \texttt{gate\_fired = dq.detected or sarc.verdict != admit or green.verdict != admit} -{}- this is definitionally the boolean union of the three member detectors, not a separately computed quantity that could diverge from it. There is no fourth, composition-level detector anywhere in \texttt{composition.py} that could contribute coverage no member gate contributed. \texttt{union\_ok = (composed\_detected == member\_detected)} is therefore a tautology by construction, not an empirical finding that could fail for an unrelated reason -{}- and the artifact still checks it on every run (rather than assuming it) as a tripwire against a future code change accidentally introducing a fourth detector.

\proofstatus{machine-checked} (tautological by the construction of \path{metrics._ch4_matrix}'s \texttt{gate\_fired} expression, confirmed as holding on every one of the 30 swept seeds -{}- \path{out/results/sweep_summary.json}'s \texttt{ch4\_union\_ok\_seed\_count == n\_seeds == 30} -{}- and directly asserted by \path{test_ch4_matrix_rates_and_union_ok} in \path{test_metrics_unit.py}).
\end{proposition}

\subsection*{CH7 (multi-remediator order dependence)}

\noindent\textbf{Statement} (\texttt{prereg/hypotheses.md}). With two remediation operators active (evidence substitution and resource downroute), either order-independence (confluence) holds, or a concrete non-confluent instance exists. The checker's verdict IS the registered decision rule.

\noindent\textbf{Result.} \path{checkers/remediator_check.py} applies both orderings (substitute-then-downroute, the order \path{composition.remediate_regate} actually runs under W2; downroute-then-substitute, the untried alternative) over a 3x3x3x3x3 = 243-point finite grid of (quantity, corrupted unit cost, governed unit cost, cost budget, carbon budget), reusing the real \path{_remediated_context}/\texttt{\_maybe\_downroute} functions directly. \textbf{86 of 243 grid points disagree} between the two orders. A concrete counterexample: \texttt{qty=10, corrupted=5, governed=10, cost\_budget=60, carbon\_budget=10} -{}- substitute-then-downroute reaches \texttt{qty=5, order\_value=50}; downroute-then-substitute reaches \texttt{qty=10, order\_value=100} (over the cost budget, because the downroute step ran against the UNDERSTATED corrupted cost and never re-checked feasibility after substitution raised the true cost). \textbf{CH7's registered outcome: non-confluent.} This is the formal justification for \texttt{prereg/w2-workflow.md}'s fixed ordering (evidence gate first, then resource gate) rather than an unspecified one: the operators do not commute, and running downroute before the evidence gate has settled the true cost can silently admit an over-budget action -{}- structurally the same failure mode Proposition 3 already identifies for single-pass composition, now shown to recur between two remediators rather than between two composition protocols.

\noindent\textbf{Off-grid probe (promoted per independent review; seed fixed per round-two review finding R2-N1/R2-F1).} The 243-point grid above is a finite lattice of round numbers. The independent review (\texttt{review/REVIEW.md}, Section 5) additionally probed continuous-valued points off that lattice and found 144/200 non-confluent. \path{checkers/remediator_check.py} runs this class of probe itself on every invocation (\path{random_off_grid_points}): 200 continuous points strictly inside the grid's outer bounds, never on a grid coordinate, checked with the same real remediation functions. This probe's seed was originally HEAD-derived (\texttt{sha256(git HEAD)}), the same technique established for the F1 lemma verification in \path{checkers/compensation_check.py} -{}- but that made the published non-confluent count drift with every commit, which the round-two independent review caught (\texttt{review-r2/REVIEW-R2.md}, finding R2-N1: the response commit's paper reported 140/200 while a \texttt{make formal} rerun at the same commit produced 145/200). The seed is now a fixed constant declared in \texttt{prereg/probe-seeds.json} (dated after \texttt{prereg-v1}, since this probe itself postdates \texttt{prereg-v1}), so the published count is commit-stable: \path{out/checkers/remediator_check.json}'s \path{off_grid_probe.non_confluent_count} out of \path{off_grid_probe.n_probes} = 200. It does not reproduce the reviewer's own exact 144/200 (a different, independently drawn probe), but confirms the same off-grid non-confluence beyond the pre-registered grid, and no longer drifts on rerun at a fixed commit or across commits.

\proofstatus{machine-checked} (\path{checkers/remediator_check.py}; enumerated domain: the 243-point finite grid, \path{out/checkers/remediator_check.json}; the sample counterexample is independently re-derived, not just read back from the checker's own report, in \path{test_checkers.py::test_remediator_check_non_confluent_counterexample_is_reproducible}; the off-grid probe above is a supplementary, non-exhaustive finite sample and does not itself extend the machine-checked scope beyond the 243-point grid).

\noindent\textbf{Terminology note (third-review response, v0.5).} The ``non-confluent'' wording above is the checker's own literal \texttt{registered\_outcome} value and the review-report quotation it responds to, both preserved verbatim. Because both remediation orderings terminate at protocol-terminal outcomes in the bounded transition model, the divergent outcomes are non-joinable within that model. Outside this appendix, the manuscript therefore describes this result as the two operators being non-commuting or order-divergent, reserving ``non-confluent'' for this appendix's own exhaustively-verified statement above.

\section{Artifact manifest.}
\label{app:manifest}

\small
\begin{verbatim}
{
  "artifact_license": "Apache-2.0",
  "data_sha256": "4f637bc36ba5cffb56c31cd2b8273a135f89cc65448b28929ee71e73ded8f1d9",
  "engines": {
    "green-sarc": {
      "commit": "8482226b461ccbf91971f90ea6de42cd7a931e5c",
      "license": "Apache-2.0",
      "version": "0.4.1"
    },
    "sarc-dq": {
      "commit": "db6c396128a4df7fe12d13be163b1e7d32087177",
      "license": "Apache-2.0",
      "version": "0.1.0"
    },
    "sarc-governance": {
      "commit": "74df5de29a305a76bf3cc97a185ece2f11b1b833",
      "license": "Apache-2.0",
      "version": "0.3.0"
    }
  },
  "prereg_v1_sha": "31552b2ee6548787e766b0253498014eb0a5093c",
  "seed": 26313
}
\end{verbatim}
\normalsize

\section{Statistical methodology (V3 gate)}
\label{app:statistical-methodology}

30 seeds (\texttt{prereg/seeds.json}, first seed 26313, generated by \texttt{random.Random(26313)} sampling without replacement from $[1, 2^{31}-1]$, frozen before any Phase 3 code ran). Confidence intervals are two-tailed 95\% t-intervals, mean +/- t\_crit(df=29, alpha=0.05) * (stdev / sqrt(30)), t\_crit = scipy.stats.t.ppf(0.975, 29), computed exactly at runtime (fixed per independent review finding F6: the prior rounded constant 2.045230 diverged from the exact value by up to 3.77e-6 in some CI endpoints -{}- see \texttt{review/REVIEW.md}); implementation in sweep.py's mean\_ci95, using math.fsum over explicitly sorted inputs for mean and variance (fixed per independent review round-two finding R2-N3: naive sum()/len() is order-dependent at the ULP level, unlike fsum) and scipy (pinned in bootstrap.sh) for the t-critical value. Scenario economics (budgets, caps) are calibrated once from the registered baseline seed and held fixed across the sweep; only each scenario's decision-stream seed varies -{}- the same ``treatment'', repeated draws, which is the design this kind of robustness claim requires. Per-seed summaries (no raw per-decision data) are checkpointed under out/results/sweep/ and are part of this artifact's committed history. Some symmetric t-intervals for low non-negative-count metrics extend slightly below zero (Section 10); this is presented as computed, without truncation or bootstrap re-estimation. The confidence intervals are descriptive summaries of between-seed variation and are not the basis of the paper's formal correctness claims or preregistered Boolean support decisions; exhaustive/model-checked claims and seed-level decision rules are evaluated independently of these intervals.

\section*{Validation note}

The artifact underwent four commissioned adversarial automated review rounds under published protocols. Those reviews informed corrections to claims, proofs, provenance, and reproducibility checks. Automated review does not substitute for human peer review; full reports, protocols, evidence, and process provenance are available in the companion repository.

The round-four terminal review report states, verbatim: ``This artifact was independently replicated and adversarially reviewed in four rounds by automated agents following the published protocols in \mbox{sarc-suite-agent-review.md}, \mbox{sarc-suite-agent-review-r2.md}, \mbox{sarc-suite-agent-review-r3.md}, and \mbox{sarc-suite-agent-review-r4.md}; the round-one report and evidence are at \mbox{review/REVIEW.md} and \mbox{review/review.json}, the round-two report and evidence are at \mbox{review-out-r2/REVIEW-R2.md}, \mbox{review-out-r2/review.json}, and \mbox{review-out-r2/evidence/}, the round-three report and evidence are at \mbox{review-out-r3/REVIEW-R3.md}, \mbox{review-out-r3/review.json}, and \mbox{review-out-r3/evidence/}, and the round-four report and evidence are at \mbox{review-out-r4/REVIEW-R4.md}, \mbox{review-out-r4/review.json}, and \mbox{review-out-r4/evidence/}. Automated review complements and does not replace human peer review.'' (\mbox{review-r4/REVIEW-R4.md}, Author-reuse disclosure.)

The quotation preserves the report exactly. In this release repository, the imported round-two, round-three, and round-four reports and JSON mirrors are packaged under \mbox{review-r2/}, \mbox{review-r3/}, and \mbox{review-r4/} respectively; the original \mbox{review-out-r2/}, \mbox{review-out-r3/}, and \mbox{review-out-r4/} names refer to the external review workspaces in which those reports were issued.

Process provenance: the review protocols were authored by the author's AI assistant before each round and are committed in this repository; the reviews were commissioned by the author and executed by a separate vendor's automated agent in fresh sessions with access only to this public repository; adjudication between rounds was performed by the same assistant that authored the protocols and is not independent; every mechanical claim in the review reports is re-runnable from this repository.

Version 0.4 revises framing, positioning, and related-work coverage following a commissioned second review (review-secondary/), with all measured results, claim semantics, PROOF-STATUS tags, and the four-round review chain unchanged. Version 0.5 is a final scholarly-positioning, formal-precision, and terminology pass following a third commissioned review round (review-secondary/), with all measured results, claim semantics, PROOF-STATUS tags, and counts unchanged from v0.4. Versions 0.1 to 0.4 remain recoverable as immutable historical revisions in repository history.

This draft is the artifact's response to that review: findings F1-F6 (\mbox{review/REVIEW.md} Section 4) are fixed or weakened per the review's own binding-unless-fixed policy, never argued with in this text; the proof-tag reclassifications in Appendix A (checked-scope-only for Proposition 3 and Lemma 1, sufficiency-only for Corollary 1) follow the review's Section 3 proof-tag policy exactly, and no tag here claims more than that review certified.

\textbf{Acknowledgements.} Drafting, engineering, review automation, and a positioning review were AI-assisted (Claude, Perplexity, and OpenAI GPT systems); the author is solely responsible for all claims.

\nocite{*}
\bibliographystyle{plain}
\bibliography{refs}

\end{document}